%% file: main.tex
\documentclass{article}
\usepackage{xcolor}
\usepackage{biblatex}
\usepackage[most]{tcolorbox}
\usepackage{braket}
\usepackage{amsmath,amssymb}
\usepackage{amsmath}
\usepackage{array}
\usepackage{mathtools}
\usepackage{amsfonts}
\usepackage{amssymb}
\usepackage{amsthm}
\usetikzlibrary{positioning, arrows.meta, calc}
\usepackage{algorithmicx}
\usepackage{braket}
\usepackage[hidelinks]{hyperref}
\usepackage{longtable}
\usepackage{physics}
\usepackage{algpseudocode}
\usepackage{booktabs}
\usepackage{algorithm}
\usepackage{float}
\usepackage{multirow}
\usepackage{tikz}
\usepackage{mdframed}
\usepackage{scrextend}
\usepackage{authblk}
\usepackage[shortlabels]{enumitem}
\usepackage[a4paper, margin=2.5cm]{geometry}
\allowdisplaybreaks
\usepackage{graphicx}

\usepackage{hyperref}

\newtcbtheorem{definition}{Def}{
	colback=black!0,
	colframe=black!50,
	sharp corners,
	separator sign={:\ }
}{def}

\newcommand{\ham}{\mathcal{H}}
\newcommand{\zsym}{\mathbb{Z}_2}
\newcommand{\vs}{\varsigma}

\newcommand{\graph}{\mathcal{G}}

\newcommand{\simp}{\sim_P}
\newcommand{\aut}{\text{aut}}
\newcommand{\autg}{\aut(\graph)}

\newcommand{\gham}{G_\ham}

\newtheorem{prop}{Proposition}

\title{Efficient Identification of Permutation Symmetries in Many-Body Hamiltonians via Graph Theory}

\author[1]{Saumya Shah \thanks{Electronic Mail: saumya.shah@u.nus.edu}}
\author[1, 2]{Patrick Rebentrost \thanks{Electronic Mail: cqtfpr@nus.edu.sg}}
\affil[1]{School of Computing, National University of Singapore, Singapore}
\affil[2]{Centre for Quantum Technologies, Singapore}
\date{}

\begin{document}

\maketitle

\begin{abstract}

The computational cost of simulating quantum many-body systems can often be reduced by taking advantage of physical symmetries.
While methods exist for specific symmetry classes, a general algorithm to find the full permutation symmetry group of an arbitrary Pauli Hamiltonian is notably lacking. This paper introduces a new method that identifies this symmetry group by establishing an isomorphism between the Hamiltonian's permutation symmetry group and the automorphism group of a coloured bipartite graph constructed from the Hamiltonian. We formally prove this isomorphism and show that for physical Hamiltonians with bounded locality and interaction degree, the resulting graph has a bounded degree, reducing the computational problem of finding the automorphism group to polynomial time. The algorithm's validity is empirically confirmed on various physical models with known symmetries. We further show that the problem of deciding whether two Hamiltonians are permutation-equivalent is polynomial-time reducible to the graph isomorphism problem using our graph representation. This work provides a general, structurally exact tool for algorithmic symmetry finding, enabling the scalable application of these symmetries to Hamiltonian simulation problems.

\end{abstract}

\section{Introduction}

The simulation of quantum many-body systems is a foundational problem in quantum chemistry and condensed matter physics. For classical computers, the exponential scaling of the Hilbert space in the number of qubits presents a computational challenge. While quantum computers offer a path forward, near-term algorithms are constrained by limited resources and noise. A primary strategy for mitigating these costs is to exploit the intrinsic physical symmetries of the target Hamiltonian.
The existence of a symmetry generally implies a conserved quantity, which partitions the system's exponentially large Hilbert space into smaller, independent subspaces, or \emph{symmetry sectors} \cite{meyerExploitingSymmetryVariational2023, gardEfficientSymmetrypreservingState2020}. Confining the simulation to a relevant sector can thus reduce the dimensionality of the problem. To exploit this technique in a scalable manner, the symmetries of the Hamiltonian must be identified algorithmically.

Several common types of symmetries are often identified. Methods for identifying discrete $\zsym$ symmetries, for example, are well-established and enable \emph{tapering}, a technique that removes qubits from the simulation by replacing their operators with eigenvalues \cite{bravyi2017taperingqubitssimulatefermionic, setiaReducingQubitRequirements2020}. Similarly, continuous symmetries like $U(1)$ and $SU(2)$ are typically identified by checking if the Hamiltonian commutes with their respective generators \cite{corry2017symmetry}.
However, a general method for identifying the full permutation symmetry group $\gham$ of an arbitrary interacting Pauli Hamiltonian is less established. The lack of a general algorithm to find it represents a significant research gap that hinders the construction of scalable permutation-equivariant circuits.

In this work, we address this gap by establishing a formal graph-theoretic isomorphism. We demonstrate that any Hamiltonian $\ham$, expressed as a weighted sum of Pauli strings, can be mapped to a faithful coloured bipartite graph $\graph$ such that the Hamiltonian's permutation symmetry group is isomorphic to the graph's automorphism group:
\begin{equation}
    \gham \cong \autg.
\end{equation}
This result reduces the problem of identifying permutation symmetries to the well-studied computational problem of finding a graph's automorphism group. We further show that for physical Hamiltonians, the resulting graphs have bounded-degree properties that permit efficient, polynomial-time solutions. We provide numerical verification of the results on common physical Hamiltonians.
The algorithmic finding of permutation symmetries discussed in this work promises the scalable construction of permutation-equivariant circuits to find ground states of many Hamiltonians.

This paper is structured as follows: Section \RN{2} surveys existing literature and presents deficiencies that motivate the method introduced in this paper. Section \RN{3} discusses the mathematical preliminaries for the method, including permutation symmetries, their mathematical formalisms, and automorphism groups. It also discusses certain physical models of interest and their known symmetry groups. Section \RN{4} introduces the explicit graph construction, proves the isomorphism between the permutation symmetry group and the automorphism group of this constructed graph, and presents a runtime analysis demonstrating polynomial-time complexity for physical systems. Section \RN{5} provides empirical validation by applying this algorithm to several physical models and verifying that the computed automorphism group generators exactly match the known, analytically-derived permutation symmetries. Section \RN{6} presents conclusions and directions for further research.

\section{Related Work}

\subsection{Symmetries in Hamiltonians}

An effective strategy to mitigate the exponential scaling of the Hilbert space in Hamiltonian simulation problems is to identify and exploit the intrinsic physical symmetries of the system's Hamiltonian. The existence of a symmetry group implies a block-diagonal structure for the Hamiltonian, partitioning the Hilbert space into independent symmetry sectors. Subsequent confinement of the simulation to a single sector thus offers a significant reduction in the problem's effective dimensionality.
The scalable application of this strategy depends on the algorithmic identification and exploitation of these symmetries. There are robust and efficient methods for several common symmetry classes. This section surveys methods in the algorithmic identification and exploitation of Hamiltonian symmetry groups.

\paragraph{Exploiting Symmetries in Hamiltonians.}
The exploitation of symmetries is a primary strategy for mitigating the computational cost of finding properties of Hamiltonians. These strategies can be broadly partitioned into several distinct technical applications.

The most direct application has been the building of \textbf{symmetry-preserving ans\"atze}, which guarantee that the variational search remains in the correct subspace. Methods for this include symmetrising the generator pool using techniques like twirling \cite{meyerExploitingSymmetryVariational2023}, building circuits that preserve particle number and other $\mathbb{Z}_N$ symmetries \cite{ayeniEfficientParticleconservingSymmetric2025}, and introducing specific symmetry-preserving gates for VQE \cite{gardEfficientSymmetrypreservingState2020}. This concept has been extended to compressing VQE ans\"atze by removing symmetry-breaking gates \cite{herasymenkoDiagrammaticApproachVariational2021}, constructing minimal complete operator pools for adaptive methods \cite{shkolnikovAvoidingSymmetryRoadblocks2023}, and building permutation-symmetric quantum neural networks \cite{schatzkiTheoreticalGuaranteesPermutationequivariant2024}. This approach is well-tested and has been shown to yield shallower circuits and reduce shot counts for adaptive methods \cite{shiraliBreakNotBreak2025, bertelsSymmetryBreakingSlows2022}, and mitigate barren plateaus \cite{schatzkiTheoreticalGuaranteesPermutationequivariant2024} in QNNs.

A complementary application of symmetry is \textbf{problem reduction}, which reduces the effective Hilbert space of the problem. The most prominent example is qubit tapering, first presented by Bravyi et al. to remove qubits based on $\mathbb{Z}_2$ symmetries \cite{bravyi2017taperingqubitssimulatefermionic}. This was later extended to molecular point group symmetries \cite{setiaReducingQubitRequirements2020}. Other reduction methods include mapping a local stoquastic Hamiltonian to one acting on a smaller symmetric subspace \cite{bringewattEffectiveGapsAre2020}, accelerating QAOA by reducing the cost Hamiltonian's Hilbert space \cite{shaydulinClassicalSymmetriesQuantum2021, shaydulinExploitingSymmetryReduces2021a, prakashAutomorphismAssistedQuantumApproximate2024}, and compressing quantum protocols into perturbative gadgets \cite{doddsPracticalDesignsPermutationsymmetric2019}.

Finally, symmetries have been exploited in other contexts, such as developing symmetry-projected bases for \textbf{efficient measurement} \cite{smartLoweringTomographyCosts2021} and developing efficient \textbf{classical simulation} algorithms using tensor networks for permutation-invariant systems \cite{anschuetzEfficientClassicalAlgorithms2023}.

\paragraph{Algorithmic Identification of Symmetries in Hamiltonians. }

Several distinct families of methods have been developed for the algorithmic identification of Hamiltonian symmetries. The most established are \textbf{algebraic and constraint-based methods}. Bravyi et al.~\cite{bravyi2017taperingqubitssimulatefermionic} present methods to algorithmically identify $\zsym$ symmetries from fermionic Hamiltonians by finding the kernel of a constructed check matrix. Setia et al.~\cite{setiaReducingQubitRequirements2020} review this method and introduce a new method to find molecular point group symmetries. Similarly, Varjas et al.~\cite{varjasQsymmAlgorithmicSymmetry2018} introduce the Qsymm framework to algorithmically find unitary symmetry generators from non-interacting Hamiltonians by solving linear commutation constraints. 

A second, related approach involves \textbf{graph-theoretic techniques}. Bringewatt et al.~\cite{bringewattEffectiveGapsAre2020} introduce techniques to map stoquastic Hamiltonians to vertex-coloured graphs to construct reduced Hamiltonians that act on a smaller symmetric subspace. Shaydulin et al.~\cite{shaydulinExploitingSymmetryReduces2021a} propose mapping the automorphism group in graph problems to permutation symmetries in its corresponding Ising cost Hamiltonian for QAOA applications. Similarly, Prakash \cite{prakashAutomorphismAssistedQuantumApproximate2024} also uses the automorphism group of the unweighted max-cut problem to partition edges into equivalence classes for cost Hamiltonian reduction. 

A recent approach is the use of \textbf{machine learning and data-driven methods}. Dierkes et al.~\cite{Dierkes_2023} introduce a Lie algebra framework that allows Hamiltonian neural networks to simultaneously learn both the symmetry group action and the total energy. Van der Ouderaa et al.~\cite{ouderaaNoethersRazorLearning2024}, for instance, present ``Noether's Razor", a Hamiltonian neural network that learns parametrised conserved quantities directly from classical trajectory data. Hou et al.~\cite{hou2024machinelearningsymmetrydiscovery} also present a framework to learn observables from classical time evolution data. 

Finally, various \textbf{numerical methods}  were established. Moudgalya et al.~\cite{moudgalyaNumericalMethodsDetecting2023} provide numerical techniques to find the commutant algebra of a Hamiltonian using simultaneous block diagonalisation and Liouvillian MPS methods. Ruzicka et al.~\cite{ruzickaConservedQuantitiesExceptional2021} introduce spectral and algebraic methods to identify conserved observables for non-Hermitian, PT-symmetric Hamiltonians. Others have developed methods to reveal symmetries that are masked by common transformations like the Jordan-Wigner transform \cite{leeuwenRevealingSymmetriesQuantum2024}.

\subsection{The Graph Automorphism Problem}\label{ssec:ga_lit}

The Graph Automorphism Problem (GA) is in the same complexity class as the Graph Isomorphism Problem (GI), as each is polynomial-time reducible to the other \cite{finkelsteinGroupsComputation1993, skresanovReductionGroupIsomorphism2025}. GI is one of the few problems in the complexity class NP that is not known to be either P (solvable in polynomial time) or NP-complete \cite{babaiGROUPGRAPHSALGORITHMS2019}. While early methods had moderately exponential bounds \cite{Luks1982, zemlyachenkoGraphIsomorphismProblem1985}, the fastest known algorithm for the general case is Babai's quasipolynomial-time algorithm \cite{babaiGraphIsomorphismQuasipolynomial2016}. This algorithm runs in time $O\left(\exp\left((\log n)^{O(1)}\right)\right)$ and remains the best-known upper bound for arbitrary graphs, thus bounding the complexity of GA in the most general case.

For special classes of graphs, however, significantly faster algorithms are known. Of particular relevance to this work is the class of \textbf{bounded-degree graphs}, where the maximum vertex degree $d$ is a constant. For these graphs, polynomial-time algorithms exist. Luks introduced a foundational algorithm based on recursive graph partitioning \cite{Luks1982} with a runtime of $O\left(n^{O(d/\log{d})}\right)$. This bound has been subsequently improved by related techniques to $O(n^{\text{polylog}(d)})$ \cite{neuenPowerAlgorithmicApproaches2019, groheFasterIsomorphismTest2023}.

In practice, GA is often solved by finding a \textbf{canonical labelling} of the graph, a problem to which GA is also reducible \cite{katebiGraphSymmetryDetection}. Modern solvers, such as Nauty \cite{nautypaper}, Bliss \cite{bliss_paper1}, and Saucy \cite{katebiGraphSymmetryDetection}, implement backtracking search trees based on individualisation-refinement algorithms \cite{stoichevNewExactHeuristic2019a}. While these practical solvers retain a worst-case exponential runtime \cite{nautypaper, stoichevNewExactHeuristic2019a}, they are known to be efficient for the majority of graphs encountered in practice.

\subsection{Unaddressed Research Gaps}

This work is motivated by primary deficiencies in existing methods.

First, the algorithmic identification of symmetries is effective for predefined, a priori known symmetry classes. Robust methods exist for discrete $\zsym$ symmetries \cite{bravyi2017taperingqubitssimulatefermionic}, molecular point group symmetries \cite{setiaReducingQubitRequirements2020}, and continuous groups like $U(1)$ and $SU(2)$, which are typically found by checking commutation with known generators. Methods for non-interacting Hamiltonians are similarly limited by their reliance on linear commutation constraints \cite{varjasQsymmAlgorithmicSymmetry2018}. A general-purpose, group-agnostic algorithm to find the full permutation group $\gham$ of an arbitrary, interacting Pauli Hamiltonian is notably lacking.

Second, existing graph-theoretic techniques are not general. These methods are specialised to specific problems such as stoquastic Hamiltonians \cite{bringewattEffectiveGapsAre2020} or specific Ising Hamiltonians used as cost functions for graph problems such as MaxCut \cite{prakashAutomorphismAssistedQuantumApproximate2024, shaydulinExploitingSymmetryReduces2021a}. A general isomorphism mapping an arbitrary Pauli Hamiltonian $\ham$ to a graph $\graph$ such that $\gham \cong \autg$ is absent from the literature.

Finally, recent numerical and machine learning methods are indirect or non-deterministic. Data-driven approaches, for example, rely on classical trajectory data to learn conserved quantities rather than deterministically identifying the symmetry group from the Hamiltonian's structure \cite{ouderaaNoethersRazorLearning2024, hou2024machinelearningsymmetrydiscovery, Dierkes_2023}. Conversely, numerical methods that compute the full commutant algebra solve a far more general and computationally expensive problem than identifying the specific permutation subgroup $\gham$.

In summary, the literature lacks a general, deterministic, and structurally exact algorithm to identify the full permutation symmetry group $\gham$ of an arbitrary Pauli Hamiltonian. This paper proposes to address this gap by introducing a novel graph-theoretic isomorphism for the algorithmic identification of the permutation symmetry group from Hamiltonians represented as weighted sums of Pauli strings. 

\section{Mathematical Preliminaries}

\subsection{Symmetries}\label{ssec:sym}

\paragraph{Hamiltonian.} A Hamiltonian is an operator that represents the total energy of a physical system. In this paper, we consider a Hamiltonian represented as 

\begin{equation}\label{eq:hamdef}
    \ham=\sum_jc_j p_j,
\end{equation}
where $c \in \mathbb{C} ~\backslash~ \{0\}$, and each $p$ belongs to the Pauli group $P_m = \pm\{I, \sigma^x, \sigma^y, \sigma^z\}^{\otimes n}$ for $n$ qubits. Since it is known that Pauli strings provide a complete basis of operators on the Hilbert space of N qubits \cite{Lawrence_2002}, this construction generalises to any qubit Hamiltonian. 
A Hamiltonian $\ham$ acting on an $n$-site system is called $k$-local if it can be expressed as a sum of terms
\begin{equation}
    \ham = \sum_iH_i,
\end{equation}
where each $H_i$ acts non-trivially on at most $k$ sites. The interaction graph $G=(V,E)$ of the Hamiltonian $\ham$ is defined by associating a vertex for each site in $V \in \{1, ..., n\}$, and an edge $(u,v) \in E$ whenever sites $u, v$ jointly appear in the support of some term $H_i$. The Hamiltonian has bounded interaction degree $d$ if the degree of every vertex in $G$ satisfies
\begin{equation}
    \deg(v) \le d\quad\forall ~v \in V,
\end{equation}
i.e., each site interacts with at most $d$ other sites.
Hamiltonians representing physical systems are known to be local \cite{lloyd1996} and have a bounded interaction degree \cite{aharonov2003adiabaticquantumstategeneration}. This is observed in common physical systems like the Ising and Heisenberg models \cite{Ising1925, Lenz1920, Heisenberg1928}.

\paragraph{Symmetries. } Let $X$ be a mathematical object. A symmetry transformation is a mapping $g: X \to X$ that preserves a relevant property of $X$, i.e., that property is invariant to $g$. The set of all transformations that preserve this property forms a \textit{group} \cite{cantwell2016introduction}.
An operator $U$ is a symmetry of Hamiltonian $\ham$ if it leaves the Hamiltonian invariant, i.e., $U\ham U^\dagger =\ham$, which is equivalent to the commutation relation \cite{corry2017symmetry}: \begin{equation}\label{eq:symcom}[\ham, U] = 0.\end{equation}

\paragraph{Permutation Symmetry.}
Permutation invariance is a discrete symmetry where the system is invariant upon the exchange of particles. For example, if a two-qubit system with state $\ket{k}\ket{k'}$ is invariant to $P_{12}$, where $P_{ij}$ corresponds to an exchange of qubits $i$ and $j$, then $\ket{k}\ket{k'}$, $\ket{k'}\ket{k}$, and, in general, any linear combination $c_1\ket{k}\ket{k'} + c_2\ket{k'}\ket{k}$ lead to an identical set of eigenvalues upon measurement \cite{Sakurai_modern_qm2010}. In this paper, we use the cyclic permutation notation, where successive elements are transposed in a cyclic manner. For example, $$(\alpha\quad\beta\quad\gamma)$$ is a $3$-cycle that corresponds to the transpositions $\alpha \to \beta, \beta \to \gamma, \gamma \to \alpha$.
Using the Hamiltonian representation given in Eq.~\eqref{eq:hamdef} with $n$ qubits and $m$ Pauli strings, we represent $\ham$ by the set of its (coefficient, Pauli string) pairs: \begin{equation}
    T = \{(c_r, p_r) : r \in [1, m]\}.
\end{equation}
We assume that for each Pauli string $p$, there is at most one pair $(c, p)$, and that each $p_r$ is chosen from $\{I, X, Y,Z\}^{\otimes n}$ (i.e., scalar phases are absorbed into the coefficient).
Define a permutation operator $\vs \in S_n$, the symmetric group on $n$ elements (representing qubits). The action of $\vs$ on a Pauli string $p = \bigotimes_{i=1}^n\sigma_\alpha^i$ is defined as:
\begin{equation}
    \vs(p) = \bigotimes_{i=1}^n\sigma_\alpha^{\vs^{-1}(i)} = \bigotimes_{j=1}^n\sigma_\alpha^{j}  \text{ on qubit $\vs(j)$}.
\end{equation}
For example, if $\vs = (1\quad3)$, then $\vs(XIY) = YIX$. Define the symmetry group $\gham$ as the set of permutations that leave $\ham$ invariant, i.e., 
\begin{equation}
    \gham = \{\vs \in S_n : \vs(\ham) = \ham\}.
\end{equation}
This is equivalent to $\vs(T) = T$, as $\vs$ is a permutation and $T$ is finite.

\subsection{Graph Automorphisms and Isomorphisms}\label{ssec:ga_def}

Let $\graph = (V, E, C_V, C_E)$ be a vertex- and edge-coloured graph as constructed in Sec. \ref{ssec:constr}. An automorphism of $\graph$ is a permutation $\phi$ of the vertex set $\phi: V \to V$ that preserves all structural properties and colouring of the graph \cite{cameronAutomorphismsGraphs2004}. Specifically, a permutation $\phi$ is an automorphism iff it satisfies the following conditions for all t $v \in V$ and all edges $(u,v) \in E$:

\begin{enumerate}[(1)]
    \item \textbf{Adjacency Preservation}: $(u, v) \in E \iff (\phi(u), \phi(v)) \in E$.
    \item \textbf{Vertex Colour Preservation}: $C_V(v)  = C_V(\phi(v))$.
    \item \textbf{Edge Colour Preservation}: $C_E((u,v)) = C_E((\phi(u),\phi(v)))$.
\end{enumerate}
The set of all such automorphisms for a graph $\graph$ forms a group under the operation of function composition, with the trivial identity permutation mapping all vertices to themselves. This group is known as the automorphism group of $\graph$, and is denoted by $\autg$ in this paper. $\autg$ is a subgroup of the symmetric group $S_{|V|}$ acting on vertices $V$ \cite{cameronAutomorphismsGraphs2004}. The computational problem of finding this group is the Graph Automorphism (GA) problem, which is known to be in the same complexity class as the Graph Isomorphism (GI) problem \cite{finkelsteinGroupsComputation1993, skresanovReductionGroupIsomorphism2025}. Some existing methods to solve this problem are discussed in Sec. \ref{ssec:ga_lit}.

More generally, two coloured bipartite graphs $\graph_1 = (V_1, E_1, C_{V_1}, C_{E_1}), \graph_2 = (V_2, E_2, C_{V_2}, C_{E_2})$ are isomorphic (i.e., $\graph_1 \cong \graph_2$ ) if there exists a bijection $\phi: V_1 \to V_2$ that preserves adjacency and colouring.

\subsection{Transverse Field Ising Model}

The Ising model was introduced by Lenz and Ising to theoretically characterise magnetic phase transitions in ferromagnetic lattice materials \cite{Lenz1920, Ising1925}. The transverse Ising model (TFIM) is a quantum analogue of the Ising model. It features a lattice with $Z_iZ_j$ nearest-neighbour interactions as well as $X_j$ terms representing interaction with an external transverse magnetic field $\Omega$ \cite{strecka2015briefaccountisingisinglike}. The TFIM is thus most generally characterised by the equation:

\begin{equation}\label{eq:gen_tfim}
    \ham = -\sum_{\langle i j\rangle} J_{ij} Z_{i} Z_{j} - \sum_j\Omega_j X_j,
\end{equation}
where $\langle ij \rangle$ denotes nearest-neighbour interactions.

\subsubsection{1D TFIM}

\paragraph{Homogenuous TFIMs.}
In the homogeneous case, all interaction strengths and fields are assumed uniform, i.e., $J_{ij}= J$, $\Omega_j = \Omega$ for all $i, j$. Therefore, the equation \eqref{eq:gen_tfim} reduces to the following equation:

\begin{equation}\label{eq:homo_tfim}
    \ham = -J\sum_{\langle i j\rangle} Z_{i} Z_{j} - \Omega \sum_j X_j.
\end{equation}
It can be shown that the 1D homogeneous TFIM with $n$ sites and periodic boundary conditions exhibits the symmetry group $D_n$ of order $2n$, which corresponds to the symmetries of a regular $n$-gon. Let the set of sites be $S =\{1, ..., n\}$. The symmetry group $\gham$ is a subgroup of the symmetry group $S_n$, which may be generated by two elements:

\begin{enumerate}[(1)]
    \item \textbf{Rotation} ($R_s$): The operation of shifting each site by $s$. $R_1: S \to S$ where $R_1(k) = (k+s)\mod n$. This corresponds to the generator 
    \begin{equation}\label{eq:gener_trans}
        (1\quad2\quad\dots\quad n)
    \end{equation}
    if $s, n$ are coprime, i.e., $\gcd(s,n) = 1$.
    \item \textbf{Reflection} ($P_r$): The operation of reflecting the system with the axis of symmetry passing through $r$. For odd $n$, $r$ corresponds to a site. For even $n$, $r$ corresponds to the midpoint of two sites. $P: S \to S$ where $P(k) = (2r-k) \mod n$. This corresponds to the generator
    \begin{equation}\label{eq:gener_refl}
        (1\quad 2r-1) (2\quad 2r-2)\dots
    \end{equation}
\end{enumerate}
It is also known that for open boundary conditions, the 1D homogeneous TFIM only exhibits the $C_2$ symmetry group, which corresponds to the cyclic group of 2. The reflection operator $P_{(n-1)/2}$ generates the symmetry group $\gham$.

\paragraph{Inhomogeneous TFIMs.}

In the inhomogeneous case, interaction strengths between different pairs of nearest neighbours may vary, possibly due to impurities \cite{chaudhari2018analyticalsolutionsinhomogeneoustransverse}. We consider the case where the coupling strength $J_{ij}$ is distinct for each nearest-neighbour interaction. Due to the distinct interactions between sites, such a system does not have any non-trivial permutation symmetries by construction.

\subsubsection{2D TFIM}

In the 2D homogeneous case, all interaction strengths and fields are assumed uniform. In addition, each site has up to four interaction terms (corresponding to the four nearest neighbours). We restrict our analysis to the simple case of the square lattice.
\begin{figure}[H]
    \centering
    \input{2d_tfim_figure}
    \caption{A square lattice corresponding to a 2D TFIM}
    \label{fig:2d_tfim}
\end{figure}
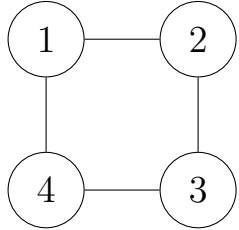
Equation \eqref{eq:gen_tfim} thus reduces to the following equation:

\begin{equation}\label{eq:2d_tfim}
    \ham=-J\left(Z_1Z_2 + Z_2Z_3 + Z_3Z_4 + Z_4Z_1\right) - \Omega\left(X_1 + X_2 + X_3 + X_4\right).
\end{equation}
It can be shown that Eq.~\eqref{eq:2d_tfim} exhibits the symmetry group $D_4$ of order 8, which corresponds to the symmetries of a square. The symmetry group $\gham$ is a subgroup of the symmetry group $S_4$ that may be generated by two elements:

\begin{enumerate}[(1)]
    \item \textbf{Horizontal Reflection} ($P_{1.5}$): The operation of reflecting the system along the vertical axis through the midpoint of sites 1 and 2. This corresponds to the generator
    \begin{equation}
        (1\quad2)(3 \quad4).
    \end{equation}
    \item \textbf{Vertical Reflection} ($P_{2.5}$): The operation of reflecting the system along the horizontal axis through the midpoint of sites 2 and 3. This corresponds to the generator
    \begin{equation}
        (1\quad4)(2 \quad3).
    \end{equation}
\end{enumerate}

\subsection{Heisenberg Model}

The Heisenberg model \cite{Heisenberg1928} is used in the study of quantum phase transitions in spin systems \cite{Joel2012}. Here, we consider a system composed of $n$ sites, which is described by the Hamiltonian:

\begin{equation}\label{eq:gen_heis_ham}
    \ham = -J\sum_{\langle i j\rangle} X_iX_j + Y_iY_j + Z_iZ_j,
\end{equation}
where $\langle i j\rangle$ denotes neighbour interactions. 
Furthermore, we consider the case where there is interaction between every pair of sites. This is also known as the mean-field Heisenberg XXX model \cite{alon2018meanfieldquantumheisenbergferromagnet}. Therefore, equation \eqref{eq:gen_heis_ham} reduces to: 

\begin{equation}\label{eq:gen_mf_heis}
    \ham = -J\sum_{i \ne j} X_iX_j + Y_iY_j + Z_iZ_j.
\end{equation}
Due to the uniform fully-connected nature of the system, it can be shown that Eq.~\eqref{eq:gen_mf_heis} exhibits the symmetry group $S_n$, which corresponds to the set of all bijective functions from a set to itself, i.e., $\gham \cong S_n$. For a set of sites $S = \{1, ..., n\}$, this group is generated by the set of adjacent transposition operators $\{T_i : i = 1, ... , n-1\}$, where each operator $T_i$ swaps sites $i$ and $i+1$, i.e.,

\begin{equation}
    T_i(k)= \begin{cases}
    i + 1 & \text{if } k = i\\
    i & \text{if } k = i + 1\\
    k & \text{Otherwise}.
    \end{cases}
\end{equation}
These generators can be expressed explicitly as $(1\quad2), (2\quad 3), ..., (n-1\quad n)$.

\section{Methodology}

\subsection{Graph Construction}\label{ssec:constr}

From the permutation operator representation described in Sec. \ref{ssec:sym}, we construct a coloured bipartite graph $\graph = (V, E, C_V, C_E)$ as follows:

\begin{itemize}
    \item $V = V_Q ~\cup~ V_T$ representing the $n$ qubits, where $V_Q = \{q_1, ..., q_n\}$ (the set of qubit nodes), and $V_T = \{t_1, ..., t_m\}$ (the set of term nodes), with a bijection $f: T \to V_T$ mapping $(c_r, p_r) \to t_r$.

    \item $E \subseteq V_Q \times V_T$. An edge $(q_i, t_r)$ exists iff the Pauli operator of $p_r$ at qubit $i$ acts non-trivially (i.e., not an identity).
    \begin{equation}
        E = \{(q_i,t_r) : [p_r]_i \ne I\}
    \end{equation}

    \item $C_V: V \to \text{Colours}_V$
    \begin{itemize}
        \item For $q_i \in V_Q$, $C_V(q_i) = \text{``qubit"}$.
        \item For $t_r \in V_E$, $C_V(t_r) = \text{(``term",$~c_r$) }$.
    \end{itemize}
    \item $C_E: E \to \{X, Y, Z\}$. For an edge $e = (q_i, t_r) \in E, ~C_E(e) = [p_r]_i$
\end{itemize}
As an example, consider the Hamiltonian with real-valued coefficients:
\begin{equation}\label{eq:exampleforgraph}
    \ham = -1 \cdot IXY + 1\cdot  YZZ - 2\cdot YXI + 2\cdot ZIZ - 3\cdot XXI + 3 \cdot YZI.
\end{equation}

\noindent Its corresponding graph is constructed as shown in Figure \ref{fig:exampleforgraph}.
\begin{figure}[t]
    \centering
    \includegraphics[width=\linewidth]{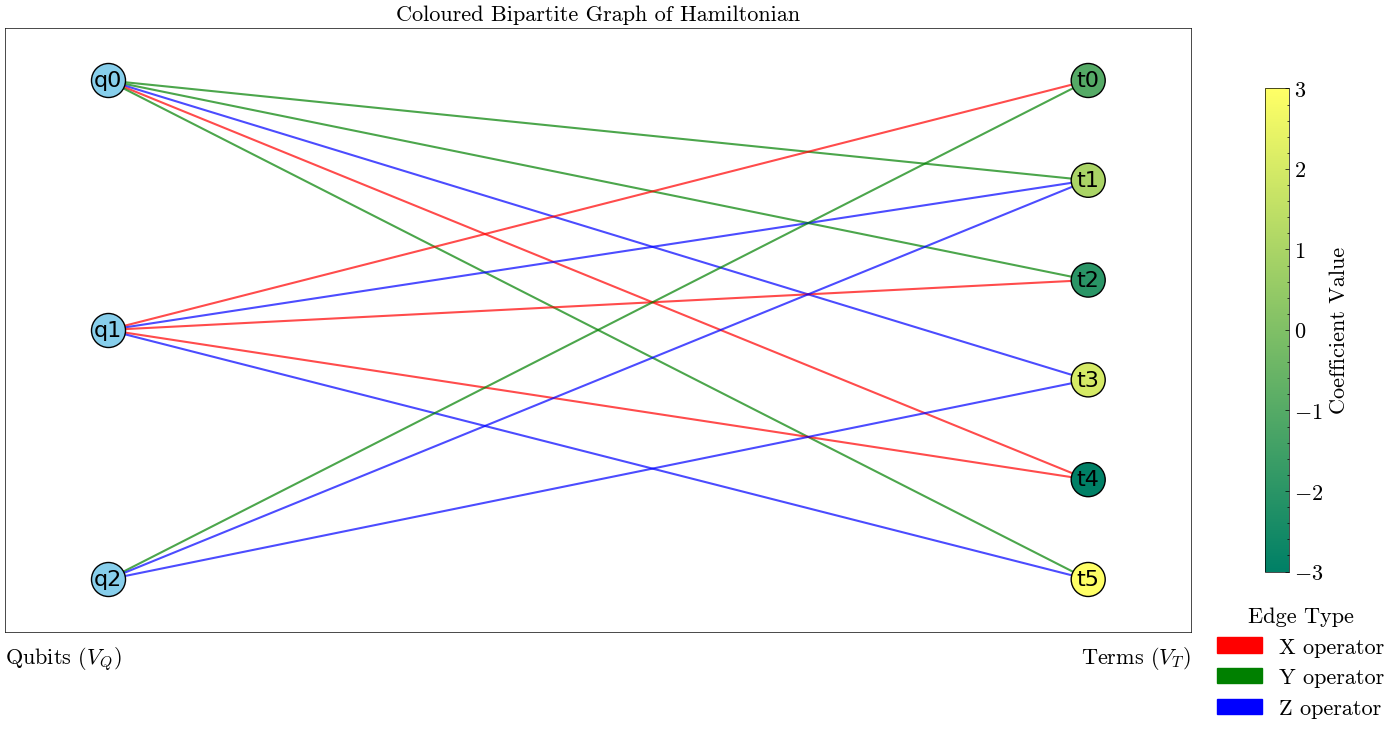}
    \caption{Coloured Bipartite Graph for Eq.~\eqref{eq:exampleforgraph}. The graph is constructed by creating nodes for the qubits and the terms in a bipartite manner, and adding edges between a term and a qubit if the term acts non-trivially on the qubit. Then, a colouring is assigned to the nodes and edges as described in Sec.~\ref{ssec:constr}.}
    \label{fig:exampleforgraph}
\end{figure}
We propose that finding the set of permutation operators is equivalent to finding the automorphism of this coloured bipartite graph, i.e.,
\begin{equation}
    \gham \cong \autg
\end{equation}
We further propose that this construction is a faithful representation, i.e., determining if two Hamiltonians are permutation-equivalent is equivalent to determining if their corresponding graphs are isomorphic. A proof of these propositions may be found in the following section \ref{ssec:proof}. Therefore, if we can compute the automorphism group of the graph $\graph$, we can find the Hamiltonian's group of permutation symmetries $\gham$. Solvers, such as Nauty or Traces \cite{nautypaper}, may be used to efficiently find the automorphism group of the constructed graph.

\subsection{Proofs of Isomorphism \& Faithful Representation}\label{ssec:proof}

In this section, we prove that finding the set of permutation operators is equivalent to finding the automorphism of the coloured bipartite graph as constructed in Sec. \ref{ssec:constr}. We also prove that the constructed graph is a faithful representation of the permutation equivalence of the Hamiltonian $\ham$.

\begin{prop}
    The permutation symmetry group of the Hamiltonian $\ham$ is isomorphic to the automorphism group of the graph $\graph$ constructed from $\ham$ as Sec. \ref{ssec:constr}.
    \begin{equation*}
        \gham \cong \autg
    \end{equation*}
\end{prop}

\begin{proof}
An automorphism of $\graph$ is a bijection $\phi: V \to V$ that preserves all graph structures and colours as detailed in Sec. \ref{ssec:ga_def}. From vertex colour preservation, any $\phi \in \autg$ must map $V_Q \rightarrow V_Q$ and $V_T \rightarrow V_T$, as their colours are distinct.

We define a map $\Psi: \gham \to \autg$ and prove it is a group isomorphism. For any $\vs \in \gham$, define $\Psi(\vs) =\phi_\vs$, where $\phi_\vs: V \to V$ is a bijection defined as:
\begin{itemize}
    \item $V_Q: \phi_\vs(q_i) = q_{\vs(i)}$
    \item $V_T$: Since $\vs \in \gham$, for every term node $t_r = (c_r, p_r) \in T$, the permuted term $t_s \coloneq (c_r, \vs(p_r))$ is also in $T$.
\end{itemize}

The map $\phi_\vs$ is a bijection, as $\vs$ is a bijection on $V_Q$ and $\vs$ action on $T$ (and thus $V_T$) is also a bijection. 

To prove that $\gham \cong \autg$, it suffices to show that:
\begin{enumerate}[(1)]
    \item $\Psi$ is well-defined, i.e., $\gham \subseteq \autg$
    \item $\Psi$ is a homomorphism, i.e., for any $\vs_1, \vs_2 \in \gham, \phi_{\vs_1 \circ \vs_2} = \phi_{\vs_1} \circ \phi_{\vs_2}$
    \item $\Psi$ is injective
    \item $\Psi$ is surjective, i.e., $\gham \supseteq \autg$
\end{enumerate}

\textbf{Part 1: $\Psi$ is well-defined, i.e., $\gham \subseteq \autg$}
\begin{addmargin}[4em]{2em}
    We show that $\phi_{\vs}$ is a group automorphism for any $\vs \in \gham$.

    \paragraph{Adjacency preservation. } An edge $(q_i, t_r)$ exists iff the Pauli operator $[p_r]_i \ne I$. Similarly, the permuted edge $(\vs(q_i), \vs(t_r))  = (q_j, t_s )$ exists iff $[p_s]_j \ne I$. By definition of the permutation action on a Pauli string $\vs(p)$, we have $[p_s]_j = [\vs(p_r)]_{\vs(i)} = [p_r]_i$. Thus, the condition $[p_r]_i \ne I$ holds iff $[p_s]_{\vs(i)} \ne I$. This demonstrates that $(q_i, t_r) \in E \iff (\vs(q_i), \vs(t_r))$, satisfying the adjacency preservation condition.

    \paragraph{Vertex colour preservation. } For any qubit node $q_i \in V_Q$, the colour $C_V(q_i) =\text{``qubit"} = C_V(q_{\vs(i)})$. For any term node $t_r \in V_T$ corresponding to $(c_r, p_r)$, the map $\phi_\vs$ yields $t_s$ corresponding to $(c_r, \vs(p_r))$. The colour $C_V(t_r) = (\text{``term"}, c_r)$ is identical to $C_V(t_s) = (\text{``term"}, c_r)$ because $\vs \in \gham$ preserves the coefficient $c_r$ of the permuted term. Thus, vertex colours are preserved for all $v \in V$.

    \paragraph{Edge colour preservation. } By construction, the edge colour of $(q_i, t_r) \in E$ is $C_E((q_i, t_r)) = [p_r]_i$. The colour of the permuted edge is $C_E((\vs(q_i), \vs(t_r))) = C_E((q_j, t_s)) = [p_s]_j$. As established in our proof of adjacency, $[p_s]_j = [p_r]_i$. Therefore, edge colours are preserved.\\

    Since $\phi_\vs$ preserves adjacency, vertex colours, and edge colours, it is a valid graph automorphism by definition. This confirms that $\Psi$ is a well-defined map $\Psi:\gham \to \autg$.
\end{addmargin}\vspace{1em}

\textbf{Part 2: $\Psi$ is a homomorphism, i.e., for any $\vs_1, \vs_2 \in \gham, \phi_{\vs_1 \circ \vs_2} = \phi_{\vs_1} \circ \phi_{\vs_2}$}
\begin{addmargin}[4em]{2em}
    Let $\vs_1, \vs_2 \in \gham$. We show that $\phi_{\vs_1 \circ \vs_2} = \phi_{\vs_1} \circ \phi_{\vs_2}$ for both sets of vertices $V_Q, V_T$.

    \paragraph{Qubit nodes. } Consider the action on an arbitrary qubit node $q_i \in V_Q$. By definition, the composition $(\phi_{\vs_1} \circ \phi_{\vs_2})(q_i)$ is evaluated as $\phi_{\vs_1}(\phi_{\vs_2}(q_i))$, which simplifies to $\phi_{\vs_1}\left(q_{\vs_1(i)}\right) = q_{\vs_1(\vs_2(i))} = q_{(\vs_1 \circ \vs_2)(i)}$, which is the definition of $\phi_{\vs_1 \circ \vs_2}(q_i)$.

    \paragraph{Term nodes. } Consider the action on an arbitrary qubit node $t_r \in V_T$ that corresponds to the pair $(c_r, p_r)$. The composite action $(\phi_{\vs_1} \circ \phi_{\vs_2})(t_r)$ is evaluated as $\phi_{\vs_1}(\phi_{\vs_2}(t_r))$. By definition, $\phi_{\vs_2}(t_r)$ maps to the node $t_s \leftrightarrow (c_r, \vs_2(p_r))$. Applying $\phi_{\vs_1}$ gives $t_u \leftrightarrow (c_r, \vs_1(p_s)) = (c_r, \vs_1(\vs_2(p_r))) = (c_r, (\vs_1\circ\vs_2)(p_r))$. This final term, by definition, is the node $t_u$ that $\phi_{\vs_1 \circ \vs_2}(t_r)$ maps to.\\

    Since the actions of $\phi_{\vs_1 \circ \vs_2}$ and $\phi_{\vs_1} \circ \phi_{ \vs_2}$ are identical on $V = V_Q \cup V_T$, the automorphisms are equal, and $\Psi$ is a group homomorphism.
\end{addmargin}\vspace{1em}

\textbf{Part 3: $\Psi$ is injective}
\begin{addmargin}[4em]{2em}
    Suppose $\Psi(\vs)$ is the identity automorphism, i.e., $\phi_\vs = \text{id}$. Then, for any $q_i \in V_Q, ~\phi_\vs(q_i) =q_i$. By definition, $\phi_\vs(q_i) = q_{\vs(i)} \implies \vs(i) = i ~\forall~i \in [1, n]$. This means that $\vs$ is the identity permutation $e \in S_n$. The kernel of $\Psi$ is trivial. Therefore, $\Psi$ is injective by definition.
\end{addmargin}\vspace{1em}

\textbf{Part 4: $\Psi$ is surjective, i.e., $\gham \supseteq \autg$}
\begin{addmargin}[4em]{2em}
     We show that for any arbitrary automorphism $\phi \in \autg$, there exists a permutation symmetry $\vs \in \gham$ such that $\Psi(\vs) = \phi$.

     Let $\phi \in \autg$ be an arbitrary automorphism. Because $\phi$ must preserve vertex colours, and $V_Q$ and $V_T$ have distinct colours, $\phi$ must map these partitions to themselves, i.e., $\phi(V_Q)=V_Q$ and $\phi(V_T)=V_T$. We can therefore define a permutation $\vs \in S_n$ based on the action of $\phi$ on the qubit vertices: $\phi(q_i)= q_{\vs(i)}$ for all $i\in[1,n]$.

     We show that this $\vs$ is a valid permutation symmetry of the Hamiltonian, i.e., $\vs \in \gham$. To this end, we show that for any term $(c_r, p_r) \in T$, the permuted term $(c_r, \vs(t_r)) \in E$.

     Consider the node $t_r \in V_T$ corresponding to $(c_r, p_r)$. By vertex colour preservation, $C_V(t_r) = C_V(t_s)$. This implies that their coefficients are identical, i.e., $c_r = c_s$, so $t_s$ corresponds to a term $(c_r,p_s) \in T$. To prove $\vs \in \gham$, we show that $p_s  = \vs(p_r)$. This equality holds iff $[p_s]_{\vs(j)} = [p_r]_j$ for all $j \in [1, n]$. We verify this by inspecting the graph structure for an arbitrary qubit index $j$:

     \paragraph{Case 1: $[p_r]_j \ne I$. } An edge $(q_j, t_r)$ exists in the graph. By adjacency preservation, an edge $(\phi(q_j), \phi(t_r)) = (q_{\vs(j)}, t_s)$ must also exist.  It is known that the colour of the edge $C_E((q_j, t_r)) = [p_r]_j$ and the colour of the permuted edge $C_E((\phi(q_j), \phi(t_r))) = [p_s]_{\vs(j)}$. By edge colour preservation, $C_E((q_j, t_r)) = C_E((\phi(q_j), \phi(t_r))) = C_E((q_{\vs(j)}, t_s))$. Therefore, $[p_r]_j =[p_s]_{\vs(j)}$.

     \paragraph{Case 2: $[p_r]_j = I$. } No edge $(q_j, t_r)$ exists in the graph. Thus, by adjacency preservation, no edge $(\phi(q_j), \phi(t_r))$ can exist. The non-existence of these edges implies $[p_s]_{\vs(j)} = I = [p_r]_j$.\\

     In both cases, we show that $[p_s]_{\vs(j)} = [p_r]_j$. Therefore, for any term $(c_r, p_r) \in T$, the automorphism $\phi$ maps its corresponding node $t_r$ to the node $t_s \leftrightarrow (c_r, \vs(p_r))$. This proves that the set of terms $T$ is invariant under the action of $\vs$, and thus $\vs \in \gham$.

     Given a $\vs \in \gham$, we confirm that $\Psi(\vs) = \phi$. By definition, $\Psi(\vs) = \phi_\vs$. The action of $\phi_\vs$ on $V_Q$ is $\phi_\vs(q_i) = q_{\vs(i)}$, which is identical to $\phi(q_i)$ by definition of $\vs$. The action on $V_T$ is $\phi_\vs(t_r) = t_s \leftrightarrow (c_r, \vs(p_r))$, which as proved, is identical to the action of $\phi(t_r)$. since $\phi$ and $\phi_\vs$ agree on all vertices $V = V_Q \cup V_T$, they are the same automorphism. Therefore, $\Psi$ is surjective.
\end{addmargin}\vspace{1em}
From the four parts above, we prove that $\gham$ is isomorphic to $\autg$
\end{proof}

Next, we show that the graph $\graph$ constructed as Sec. \ref{ssec:constr} provides a faithful representation of the permutation equivalence of the Hamiltonian $\ham$. Two $n$-qubit Hamiltonians $\ham_1, \ham_2$ represented by their sets of (coefficient, Pauli string) pairs $T_1, T_2$ are permutation equivalent iff there exists a permutation $\vs \in S_n$, the symmetric group on $n$ qubits, such that the action of $\vs$ on $T_1$ yields $T_2$. The mathematical statement is
\begin{equation}
    \ham_1 \simp \ham_2 \iff \exists~ \vs \in S_n \text{  such that } \vs(T_1) = T_2.
\end{equation}

\begin{prop}
The graph $\graph$ constructed in Sec. \ref{ssec:constr} is a faithful representation of the permutation equivalence of the Hamiltonian $\ham$, i.e.,
\begin{equation*}
\ham_1 \sim_P\ham_2 \iff \graph_1 \cong \graph_2.
\end{equation*}
\end{prop}

\begin{proof} We prove sufficiency ($\implies$) and necessity ($\impliedby$) separately.\vspace{1em}

\textbf{Part 1: $\ham_1 \sim_P\ham_2 \implies  \graph_1 \cong \graph_2$}
\begin{addmargin}[4em]{2em}
    Assume that $\ham_1, \ham_2$ are permutation-equivalent given a permutation $\vs$. We construct a graph isomorphism $\phi: V_1 \to V_2$ from $\vs$. Define the bijection $\phi$ on the vertex set $V = V_Q \cup V_T$ as:
    \begin{itemize}[\label{}]
        \item \textbf{$V_Q$:} $\phi$ applies the permutation $\vs$, i.e., for any $q_i \in V_{Q,1},~ \phi(q_i) = q_{\vs(i)} \in V_{Q,2}$.
        \item \textbf{$V_T$:} Since $\vs(T_1) = T_2$ by definition, it is known that for every term node $t_r \leftrightarrow (c_r, p_r) \in T_1$, there exists $t_s \leftrightarrow (c_r, \vs(p_r)) \in T_2$. Thus, $\phi(t_r) = t_s$, where $t_s \leftrightarrow (c_r, \vs(p_r))$.
    \end{itemize}
We verify the conditions for graph isomorphism:
    \paragraph{Adjacency preservation. } An edge $(q_i, t_r) \in E_1$ exists iff $[p_r]_i \ne I$. The image of this edge is $(\phi(q_i), \phi(t_r)) = (q_{\vs(i)}, t_s)$. An edge exists here if $[\vs(p_r)]_{\vs(i)} \ne I$. By definition of the permutation operation on Pauli strings, we have $[p_r]_i = [\vs(p_r)]_{\vs(i)}$. Thus, $(q_i, t_r) \in E_1 \iff (\phi(q_i), \phi(t_r)) \in E_2$.

    \paragraph{Vertex colour preservation. } For qubit nodes, there is no change in colours, i.e., for $q_i \in V_{Q, 1}$, $C_{V,1}(q_i) = C_{V,2}(q_{\vs(i)}) = \text{``qubit"}$. For term nodes $t_r$, $C_{V, 1}(t_r) = (\text{``term"}, c_r).$ Since $\vs$ only permutes the qubits within the Pauli string, the coefficient is preserved, i.e., $C_{V, 2}(t_s) = C_{V, 1}(t_r) = (\text{``term"}, c_r)$. Therefore, vertex colours are preserved.

    \paragraph{Edge colour preservation. } The colour of the original edge is $C_{E,1}((q_i, t_r)) = [p_r]_i$, and the colour of the permuted edge is $C_{E,2}((q_{\vs(i)}, t_s)) = [\vs(p_r)]_{\vs(i)}$. As established, the operators are identical under the permutation operation on the Pauli string, i.e., $[p_r]_i = [\vs(p_r)]_{\vs(i)}$. Therefore, edge colours are preserved.\\

    Since $\phi: V_1 \to V_2$ preserves the required properties, $\phi$ is an isomorphism, and thus, $\graph_1 \cong \graph_2$.
\end{addmargin}\vspace{1em}

\textbf{Part 2: $\ham_1 \sim_P\ham_2 \impliedby  \graph_1 \cong \graph_2$}
\begin{addmargin}[4em]{2em}
    Assume that $\graph_1, \graph_2$ are isomorphic given a mapping $\phi: V_1\to V_2$. We construct a permutation $\vs \in S_n$ from the automorphism $\phi$.
    Because the vertex colouring distinguishes qubit nodes $V_Q$ from term nodes $V_T$, the automorphism $\phi$ maps $V_{Q, 1} \to V_{Q,2}$ and $V_{T,1} \to V_{T,2}$. Define the permutation $\vs \in S_n$ from the action of $\phi$ on the qubit vertices:

    \begin{equation}
        \phi(q_i) = q_{\vs(i)} ~\forall~i \in [1, n].
    \end{equation}
    We show that the action of $\vs$ on $T_1$ yields $T_2$, i.e., $\vs(T_1) = T_2$. Consider a term node $t_r \in V_{T,1}$ corresponding to $(c_r, p_r) \in T_1$, and its image node $\phi(t_r) = t_s \in V_{T,2}$ corresponding to $(c_s, p_s) \in T_2$.

    \paragraph{Coefficient preservation. } By vertex colour preservation, $C_{V,1}(t_r) = C_{V,2}(t_s)$. Since the colour of a term node includes it coefficient, this implies that $c_r = c_s$.

    \paragraph{Pauli string equivalence. } To prove that $p_s=\vs(p_r)$, we show that the two strings have identical operators on the qubits, i.e., $[p_s]_{\vs(j)}=[p_r]_j ~ \forall ~ j \in [1,n]$ qubit indices.

    \begin{itemize}[\label{}]
        \item \textbf{Case $[p_r]_j \ne I$}: An edge $(q_j, t_r)$ exists. By adjacency preservation, $(\phi(q_j), \phi(t_r)) = (q_{\vs(j)},t_s)$ must also exist. By edge colour preservation, $C_{E,1}((q_j, t_r)) =C_{E,2}((q_{\vs(j)},t_s))$, which is true only if $[p_r]_j = [p_s]_{\vs(j)}$.
        \item \textbf{Case $[p_r]_j = I$}: No edge $(q_j, t_r)$ exists. By adjacency preservation, $(\phi(q_j), \phi(t_r)) = (q_{\vs(j)},t_s)$ cannot exist, implying $[p_r]_j = [p_s]_{\vs(j)} = I$.
    \end{itemize}
    In both cases, we show $[p_r]_j = [p_s]_{\vs(j)}$, i.e., $p_s = \vs(p_r)$. Since the isomorphism $\phi$ maps every term $(c_r, p_r) \in T_1$ to a term $\phi(t_r) \leftrightarrow (c_r, \vs(p_r)) \in T_2$, the set of terms $T_2$ is equivalent to the set of terms $T_1$ under the action of permutation $\vs$. Thus, $\vs(T_1)= T_2$, and $\ham_1 \simp H_2$. 
\end{addmargin}\vspace{1em}

From the two parts above, we prove that the graph $\graph$ is a faithful representation of the permutation equivalence of the Hamiltonian $\ham$.
\end{proof}

\subsection{Degree bound}

Hamiltonians representing physical systems often have a constant locality $k$ and a bounded interaction degree $d$. From these properties, it can be shown that their graphs have a bounded degree that does not scale with system size.

\begin{prop}
    The degree of the constructed graph $\graph$ for a Hamiltonian $\ham$ with constant locality $k$ and bounded interaction degree $d$ is bounded, i.e., \begin{equation}
        \deg(\graph) \in O(1).
    \end{equation}
\end{prop}

\begin{proof}
The graph $\graph$ has vertex set $V = V_Q \cup V_T$, where $V_Q$ are qubit vertices and $V_T$ are term vertices corresponding to Pauli terms of the Hamiltonian.
For each term vertex $v_T \in V_T$, 
\begin{equation}\label{eq:lim_degvt}
\deg(v_T) \leq k,
\end{equation}
since each term acts non-trivially on at most $k$ qubits by the locality of the Hamiltonian.
Consider a qubit vertex $v_Q \in V_Q$. The degree $\deg(v_Q)$ is the number of terms in $\ham$ that act non-trivially on qubit $v_Q$.
Since the Hamiltonian has an interaction range bounded by $d$, any Pauli term acting non-trivially on qubit $v_Q$ must have its support within a neighbourhood of size $d$ around $v_Q$ in the interaction graph of the Hamiltonian.

Within this neighbourhood of size at most $d$, we select subsets of $k-1$ qubits (excluding itself) to form a support of size $k$ including $v_Q$. The number of such subsets is bounded by the binomial coefficient $$\binom{d}{k-1}.$$
For each such choice of support, there are $3^k$ possible non-trivial Pauli strings, since each qubit operator can be $X$, $Y$, or $Z$.
Therefore, the total number of terms acting non-trivially on $v_Q$ is bounded by \begin{equation}\label{eq:lim_degvq}
\deg(v_Q) \leq \binom{d}{k-1} \times 3^k.
\end{equation}
Therefore, the degree of the graph is
\begin{align*}
\deg(\graph) &= \max_{v \in V}\left\{\max_{v_Q \in V_Q}\deg(v_Q), \max_{v_T \in V_T}\deg(v_T)\right\}\\
\deg(v_T) &\in O(k) & \text{[From Eq.~\eqref{eq:lim_degvt}]}\\
\deg(v_Q) &\in O(3^k) & \text{[From Eq.~\eqref{eq:lim_degvq}]}\\
\implies \max_{v_T \in V_T}\deg(v_T) &\le \max_{v_Q \in V_Q}\deg(v_Q)\\
\implies \deg(\graph) &\in O(3^k)\\
\therefore \deg(\graph)&\in O(1) & \text{[Since $k, d \in O(1)$]}.
\end{align*}
\end{proof}
This result enables the use of Luks' \cite{Luks1982} and Neuen's \cite{neuenPowerAlgorithmicApproaches2019} methods for polynomial-time graph automorphism finding, and thus, permutation symmetry group finding on Hamiltonians representing physical systems.

\section{Validation}

To validate the claim that the Hamiltonian's permutation symmetry group $\gham$ is isomorphic to the automorphism group of its corresponding coloured bipartite graph $\autg$, we performed a series of empirical tests.
We implemented the graph construction algorithm as detailed in Sec. \ref{ssec:constr} in Python. This algorithm converts a Hamiltonian $\ham$, represented as a dictionary of (Pauli string, coefficient) pairs, into the coloured bipartite graph $\graph$. The generators of the automorphism group $\autg$ were then computed using standard computational graph theory libraries (NetworkX \cite{networkx_paper} and Nauty \cite{nautypaper}).
We apply this implementation to a set of physical models for which the permutation symmetries $\gham$ are analytically known. Specifically, we consider varying cases of the Ising and Heisenberg models. We discard the trivial permutation symmetry for every experiment. A successful validation requires that the computed $\autg$ is isomorphic to the known $\gham$ for all cases.

\subsection{Transverse Field Ising Model}

In this section, we present and discuss the results of applying our algorithm to various cases of the transverse field Ising model.

\begin{table}[H]
  \centering 
  \caption{Known and Found Symmetry Groups for 1D Homogeneous TFIM (Periodic Boundary Conditions) with Varying System Size}
  \label{tab:h_tfim_pbc_results}
  \begin{tabular}{@{} c c m{7cm} @{}}
    \toprule
    Number of Qubits (n) & Known Symmetry Group $\gham$ & Generators of $\autg$ Found\\
    \midrule
    5 & $D_5$ & (2\quad5)(3\quad4),\newline (1\quad2)(3\quad5)\\
    \midrule
    8 & $D_8$ & (2\quad8)(3\quad7)(4\quad6),\newline(1\quad2)(3\quad8)(4\quad7)(5\quad6)\\
    \midrule
    10 & $D_{10}$ &
    (2\quad10)(3\quad9)(4\quad8)(5\quad7),\newline(1\quad 2)(3\quad10)(4 \quad9)(5\quad8)(6\quad7)
    \\
    \bottomrule
  \end{tabular}
\end{table}

\paragraph{1D Homogeneous TFIM with Periodic Boundary Conditions. } We observe that the first generator in all cases in table \ref{tab:h_tfim_pbc_results} corresponds to the reflection $P_1$. In the first case, the second generator corresponds to $P_4$. The composition of $P_4$ and $P_1$ yields the rotation operator $R_1$. This can be shown as follows:
\begin{align*}
    P_4(P_1(k)) &= P_4(2- k\pmod5) \\
    &= (3-(2-k))\mod 5\\
    &= (1+k)\mod 5\\
    &=R_1(k).
\end{align*}
Since both $P_1$ and $R_1$ can be generated from the generators found in the first case, the algorithm correctly identifies the $D_5$ symmetry in the first case. For the second and third cases with generators $\{P_1, P_{5.5}\}$ and $\{P_1, P_{6.5}\}$, the operator $R_1$ may be generated by $P_{5.5} \circ P_1$ and $P_{6.5} \circ P_1$. Since the rotation step $s = 1$ is co-prime to 8 and 10, the $D_8$ and $D_{10}$ symmetry groups are generated by the found generators.
Therefore, the algorithm correctly identifies the known $\gham \cong D_n$ symmetry group in all cases of the 1D homogeneous TFIM with periodic boundary conditions.

\begin{table}[H]
  \centering 
  \caption{Known and Found Symmetry Groups for 1D Homogeneous TFIM (Open Boundary Conditions) with Varying System Size}
  \label{tab:h_tfim_obc_results}
  \begin{tabular}{@{} c c m{7cm} @{}}
    \toprule
    Number of Qubits (n) & Known Symmetry Group $\gham$ & Generators of $\autg$ Found\\
    \midrule
    5 & $C_2$ &(1\quad5)(2\quad4)\\
    \midrule
    8 & $C_2$ & (1\quad8)(2\quad7)(3\quad6)(4\quad5)\\
    \midrule
    10 & $C_{2}$ &
    (1\quad10)(2\quad9)(3\quad8)(4\quad7)(5\quad6)\\
    \bottomrule
  \end{tabular}
\end{table}

\paragraph{1D Homogeneous TFIM with Open Boundary Conditions. }

We observe that the found generators $P_{3}, P_{4.5}, P_{5.5}$ in table \ref{tab:h_tfim_obc_results} correspond exactly to the required $P_{(n+1)/2}$ generator in all cases, thereby generating the correct $C_2$ group. Therefore, the algorithm correctly identifies the known $\gham = C_2$ symmetry group in all cases of the 1D homogeneous TFIM with open boundary conditions.

\begin{table}[H]
  \centering 
  \caption{Known and Found Symmetry Groups for 1D Inhomogeneous TFIM with Varying System Size}
  \label{tab:ih_tfim_results}
  \begin{tabular}{@{} c c m{7cm} @{}}
    \toprule
    Number of Qubits (n) & Known Symmetry Group $\gham$ & Generators of $\autg$ Found\\
    \midrule
    5 & Nil & Nil\\
    \midrule
    8 & Nil & Nil\\
    \midrule
    10 & Nil & Nil\\
    \bottomrule
  \end{tabular}
\end{table}

\paragraph{1D Inhomogeneous TFIM}
We observe that no generators are found in table \ref{tab:ih_tfim_results} as expected.

\begin{table}[H]
  \centering 
  \caption{Known and Found Symmetry Groups for 2D TFIM (Square Lattice)}
  \label{tab:sq_tfim_results}
  \begin{tabular}{@{} c c m{7cm} @{}}
    \toprule
    Number of Qubits (n) & Known Symmetry Group $\gham$ & Generators of $\autg$ Found\\
    \midrule
    4 & $D_4$ & (2\quad4),\newline(1\quad2)(3\quad4)\\
    \bottomrule
  \end{tabular}
\end{table}

\paragraph{2D Homogeneous TFIM} We observe that the second generator in table \ref{tab:sq_tfim_results} corresponds to the horizontal $P_{1.5}$ rotation operator. We can show that composing this operator with the first generator, i.e., $P_{1.5} ~\circ ~(2\quad 4)$ corresponds exactly to the vertical reflection operator $P_{2.5}$. Therefore, the algorithm correctly identifies the $\gham \cong D_4$ symmetry group in the case of the 2D square lattice homogeneous TFIM.

\subsection{Heisenberg Model}

\begin{table}[H]
  \centering 
  \caption{Known and Found Symmetry Groups for 1D Mean-Field Heisenberg XXX Model}
  \label{tab:heis_xxx_results}
  \begin{tabular}{@{} c c m{7cm} @{}}
    \toprule
    Number of Qubits (n) & Known Symmetry Group $\gham$ & Generators of $\autg$ Found\\
    \midrule
    5 & $S_5$ &(1\quad2), (2\quad3),\newline(3\quad4), (4\quad5)\\
    \midrule
    8 & $S_8$&(1\quad2), (2\quad3),\newline(3\quad4), (4\quad5),\newline(5\quad6), (6\quad7),\newline(7\quad8)\\
    \midrule
    10 & $S_{10}$&(1\quad2), (2\quad3), (3\quad4),\newline(4\quad5), (5\quad6), (6\quad7),\newline(7\quad8), (8\quad9), (9\quad10)\\
    \bottomrule
  \end{tabular}
\end{table}

\paragraph{1D Mean-Field Heisenberg XXX Model. } We observe that the generators in each case correspond exactly to the set of expected generators $\{T_i\}$. Therefore, the algorithm correctly identifies the $\gham = S_n$ symmetry group in all cases of the 1D Mean-Field Heisenberg XXX Model.

\section{Conclusion}
\label{sec:conclusion}

This work addresses a notable lack of a general, deterministic, and structurally exact algorithm to identify the full permutation symmetry group ($\gham$) of an arbitrary Hamiltonian $\ham$ expressed as a weighted sum of Pauli strings. We have introduced a novel method that establishes an isomorphism between the Hamiltonian's permutation symmetry group and the automorphism group of a corresponding coloured bipartite graph $\graph$ as constructed in Sec. \ref{ssec:constr}. The isomorphism $\gham \cong \autg$ was proven in Sec. \ref{ssec:proof} by demonstrating that the mapping $\Psi: \gham \to \autg$ is a well-defined group isomorphism. Furthermore, the proof of faithful representation establishes that the problem of deciding whether two Hamiltonians are permutation-equivalent is polynomial-time reducible to the Graph Isomorphism (GI) problem.

This result has significant computational implications for the simulation of Hamiltonians of physical systems. We proved that for physical Hamiltonians, which are characterised by bounded locality $k$ and a bounded interaction degree $d$, the constructed graph $\graph$ necessarily has a bounded vertex degree ($\deg(\graph) \in O(1)$). This result reduces the complexity of finding $\autg$ from the quasi-polynomial time required for general graphs to polynomial time using established algorithms for bounded-degree graphs. Our method thus provides a computationally efficient and scalable tool for a problem previously lacking a general solution. The validity of this graph-theoretic isomorphism was empirically confirmed by applying the algorithm to several physical models with analytically known symmetry groups, such as various 1D and 2D Transverse Field Ising Models and the mean-field Heisenberg XXX model, and verifying in all cases that the generators of the computed $\autg$ correctly reproduced the known $\gham$.

The algorithm's performance must be validated through large-scale benchmarking against existing numerical and constraint-based methods. Theoretical extensions should be explored, such as adapting the graph construction to directly identify symmetries in fermionic Hamiltonians prior to transformations (like the Jordan-Wigner transform) that may obscure them. Finally, investigating how small, symmetry-breaking perturbations in $\ham$ are reflected in the structure of $\graph$ could provide a pathway to identifying approximate symmetries by relaxing the conditions of the exact graph automorphism method presented in this paper.

\subsection*{Acknowledgements}

We express gratitude to Long Tianqi for his insights on proofs in graph theory and group theory. We also thank Clyde Lhui Kay Kin and Fan Yueh-Ching for their helpful discussions.

\newpage
\printbibliography
\newpage

\end{document}

%% file: 2d_tfim_figure.tex
\begin{tikzpicture}[
    mycircle/.style = {
        draw,
        circle,
        minimum size = 1cm, 
        font = \Large
    }
  ]

  \node[mycircle] (1) at (0, 2) {1};
  
  \node[mycircle] (2) at (2, 2) {2};
  
  \node[mycircle] (3) at (2, 0) {3};
  
  \node[mycircle] (4) at (0, 0) {4};

  \draw (1) -- (2); 
  \draw (3) -- (4); 
  \draw (1) -- (4); 
  \draw (2) -- (3); 

\end{tikzpicture}